\documentclass[runningheads]{llncs}
\usepackage{algorithm,algpseudocode}
\usepackage{mathtools}
\DeclarePairedDelimiter\ceil{\lceil}{\rceil}

\usepackage{graphicx}
\usepackage{tabularx}
\usepackage{textcomp}
\usepackage{xcolor}
\usepackage[labelformat=simple]{subcaption}

\usepackage{booktabs}

\usepackage{pgf,tikz,pgfplots}
\pgfplotsset{compat=1.15}
\usepackage{mathrsfs}
\usepackage{hyperref}
\usepackage{multicol}

\usepackage{xcolor}
\usepackage[linesnumbered,ruled,vlined,algo2e]{algorithm2e}

\SetCommentSty{mycommfont}

\begin{document}
\title{Filling MIS Vertices of a Graph by Myopic Luminous Robots
}

\titlerunning{Filling MIS Vertices of a Graph by Myopic Luminous Robots}
\author{}\institute{}

\author{Subhajit Pramanick\inst{1} \and Sai Vamshi Samala\inst{1} \and  Debasish Pattanayak\inst{2} \and Partha Sarathi Mandal\inst{1}\thanks{psm@iitg.ac.in}}
\institute{Department of Mathematics \\ Indian Institute of Technology Guwahati, India
\and
LUISS University, Rome,  Italy} 

\authorrunning{Pramanick et al.}

\maketitle
\begin{abstract}

    We present the problem of finding a maximal independent set (MIS) (named as \emph{MIS Filling problem}) of an arbitrary connected graph with luminous myopic mobile robots. The robots enter the graph one after another from a particular vertex called the \emph{Door} and move along the edges of the graph without collision to occupy vertices such that the set of occupied vertices form a maximal independent set. 

    In this paper, we explore two versions of the \textit{MIS filling problem}. For the \emph{MIS Filling with Single Door} case,  our IND algorithm forms an MIS of size $m$ in $O(m^2)$ epochs under an asynchronous scheduler, where an epoch is the smallest time interval in which each participating robot gets activated and executes the algorithm at least once.
    The robots have three hops of visibility range, $\Delta + 8$ number of colors, and $O(\log \Delta)$ bits of persistent storage, where $\Delta$ is the maximum degree of the graph. For the \emph{MIS Filling with Multiple Doors} case, our MULTIND algorithm forms an MIS in $O(m^2)$ epochs under a semi-synchronous scheduler using robots with five hops of visibility range, $\Delta + k + 7$ number of colors, and $O(\log (\Delta + k))$ bits of persistent storage, where $k$ is the number of doors.
\end{abstract}


\section{Introduction}\label{sec:appendix}
\subsection{Motivation}
The coordination among large number of autonomous mobile robots or agents has gained significant interest in recent years. Under the framework of ``Look-Compute-Move" cycles, the robots can perform various tasks such as exploration \cite{albers2000exploring}, gathering \cite{d2016gathering,defago2020self,kamei2019gathering,peleg2005distributed}, pattern formation \cite{bose2020arbitrary,SuzukiY99}, dispersion \cite{augustine2018dispersion,kshemkalyani2019efficient,sharma2016complete}, scattering \cite{PoudelS19,hideg2018filling,PoudelS20} and others.
In this work, we consider the underlying environment as a graph, and the robots can stay at the nodes and move along the edges.

In general, maximal independent set (MIS) of a network graph play a significant role in decomposing the network into clusters of low diameter, which is often very useful in designing and implementing distributed divide and conquer algorithms. MIS vertices can also be used as a network backbone for deploying communication infrastructure. For example, information dissemination in a low latency system where all robots form a network should be located at MIS vertices so that all other vertices are just one hop away.


The Filling problem, introduced by Hsiang et al. ~\cite{hsiang2004algorithms}, considers the robots enter via particular vertices and fill an environment (graph) composed of pixels (vertices) and robots occupy every pixel (vertex).
Later Hideg et al.~\cite{hideg2020asynchronous} presented the Filling problem for an arbitrary connected graph.
It is of interest to cover the entire graph but using a smaller number of robots. Thus forming an MIS by the robots that enter the graph becomes a natural extension.
We call this problem the \emph{MIS Filling problem}.

In this paper, we consider \textit{luminous robots}, that are mobile robots possessing externally visible persistent memory (or \textit{lights}).
Each vertex can contain at most one robot at a time.
We say a collision happens when two or more robots move to the same vertex.
Only one robot can travel along one edge at a time.
In this problem, the robots enter the graph one by one through a specific vertex called the \textit{Door} and move in the graph along the edges from one vertex to another while avoiding a collision.
The objective is only to occupy vertices that form an MIS.
We solve two flavors of the problem: graphs with a single Door under an asynchronous (ASYNC) scheduler and graphs with multiple Doors under a semi-synchronous (SSYNC) scheduler.
We use \textit{epochs} to denote the time complexity, where an epoch is the smallest amount of time required for all the participating robots to activate once.
On each activation a robot executes a \textit{Look-Compute-Move} \textit{(LCM)} cycle.
In ASYNC, the cycles are independently executed within finite but unpredictable time.
In SSYNC, time is discretely separated into rounds, and a subset of the robots are activated in each round and finish the execution of a cycle in the same round.
Having multiple Doors instead of just one offers redundancy in situations where a Door can be blocked.

\subsection{Related works}

Kamei and Tixeuil \cite{kamei2020asynchronous}  solve two variations of the maximum independent set (MAX\_IS) placement problem for grid networks. The first one assumes knowledge of port-numbering for each node. It uses three colors of light and a visibility range of two. 
The other one removes the assumption of port-numbering and uses seven colors of light and a visibility range of three. Barrameda et al.
\cite{10.1007/978-3-540-92862-1_11} proposed algorithms for uniform dispersal or filling problem on any simply connected orthogonal space using identical asynchronous sensors. They present two algorithms; one for the single door, where sensors have one unit of visibility range and two-bit of persistent memory, and the other for multiple doors, where sensors have two units of visibility and a constant amount of persistent memory. They also prove that oblivious sensors cannot solve the problem deterministically even if they have unlimited visibility. For multiple doors, they show that the problem is unsolvable if the visibility range is less than two, even if sensors have unbounded memory. Further, even with unbounded visibility and memory, they show that the problem is unsolvable if the sensors are identical. Barrameda et al.  \cite{10.1007/978-3-642-45346-5_17} extended the problem of uniform dispersal for orthogonal domains with holes. They solve the problem when robots have a visibility range of six without any direct communication among themselves. Later, they solve the problem using direct communication among robots to reduce the visibility radius without increasing the memory requirement. Hideg and Lukovszki  \cite{hideg2017uniform} solve the filling problem in orthogonal regions, where the robots enter the region through entry points, called doors. They propose two algorithms with run-time $O(n)$, one for single door and the other for multiple door case. Later Hideg and Lukovszki \cite{hideg2020asynchronous} presented the Filling problem for an arbitrary connected graphs in asynchronous setting where the goal is to fill the entire graph using myopic luminous robots. 

The algorithm proposed by Hideg and Lukovszki \cite{hideg2020asynchronous} cannot be directly applied in \emph{MIS Filling problem}, as the communication process in their PACK algorithm designed around one hop movement of the robots. For starters, one needs to maintain a two hop gap between the chain of robots, while simultaneously ensuring that the chain never crosses itself. The chain crossing problem does not arise if the chain is closely packed at one hop distance, and it becomes challenging in the presence of multiple chains. Also, at no point we can allow more than the size of MIS of robots to enter the graph, since that would render the problem unsolvable. We show that our algorithms handle these additional requirements and correctly form an MIS.
The state-of-the-art results and ours' are given in Table~\ref{comp_tab1_app}, where $m$ is the number of robots that form MIS.

\vspace{-1.13cm}
\begin{table}[ht]

\caption{The state of the art of previous results PACK \cite{hideg2020asynchronous}, BLOCK \cite{hideg2020asynchronous} and MIS placement on grid \cite{kamei2020asynchronous} with our proposed algorithms IND and MULTIND. }
    \label{comp_tab1_app}
    \begin{tabular}{|p{6em}|p{5em}|p{4.5em}|p{5em}|p{5em}|p{4.6em}|p{4.7em}|}
        \hline 
        \textbf{Algorithm}         & PACK \cite{hideg2020asynchronous} & BLOCK \cite{hideg2020asynchronous} & Algo 1  \cite{kamei2020asynchronous} & Algo 2 \cite{kamei2020asynchronous} & IND                   & MULTIND                     \\
        \hline
        \textbf{Scheduler}         & ASYNC                             & ASYNC                              & ASYNC                                & ASYNC                               & ASYNC                 & SSYNC                       \\
        \hline
        \textbf{Problem}           & Filling Problem                   & Filling Problem                    & MAX\_IS Placement                    & MAX\_IS Placement                   & \emph{MIS Filling}    & \emph{MIS Filling}          \\
        \hline
        \textbf{Topology}          & Connected graph                   & Connected graph                    & Grid network                         & Grid network                        & Connected graph       & Connected graph             \\
        \hline
        \textbf{Number of Doors}   & Single                            & Multiple ($k$)                     & Single                               & Single                              & Single                & Multiple ($k$)              \\
        \hline
        \textbf{Visibility Range}  & 1 hop                             & 2 hops                             & 2 hops                               & 3 hops                              & 3 hops                & 5 hops                      \\
        \hline
        \textbf{Memory (in bits)} & $O(\log \Delta)$             & $O(\log \Delta)$              & $O(1)$                         & $O(1)$                         & $O(\log \Delta)$ & $O(\log (\Delta + k))$  \\
        \hline
        \textbf{Number of Colors}  & $\Delta + 4$                      & $\Delta + k + 4$                   & 3                                    & 7                                   & $\Delta + 8$          & $\Delta + k + 7$            \\
        \hline
        \textbf{Time Complexity}   & $O$($n^{2}$) epochs               & $O$($n$) epochs                    & $O(n(L+l))$ steps\footnotemark       & $O(n(L+l))$ steps                   & $O$($m^{2}$) epochs   & $O$($m^{2}$) epochs         \\
        \hline
    \end{tabular}
    \vspace{\baselineskip}
    
\end{table}\vspace{-1.14cm}

\footnotetext{The steps represent the total movement of robots throughout the execution of the algorithm. The number of nodes and grid dimensions are represented by $n$, $L$ and $l$, respectively.}

\subsection{Contributions}
In this paper, we propose two algorithms IND and MULTIND corresponding to single and multiple doors.
\begin{itemize}
    \item  
    Algorithm IND solves the \emph{MIS Filling problem} in graphs with a single Door under an ASYNC scheduler using robots with a visibility range of 3, $\Delta + 8$ number of colors, $O(\log \Delta)$ bits of persistent storage in $O(m^{2})$ epochs, where $\Delta$ is the maximum degree of the graph, and $m$ is the number of robots that form an MIS.
    \item Algorithm MULTIND solves the \emph{MIS Filling problem} in graphs with multiple Doors under an SSYNC scheduler using robots with visibility range of 5, $\Delta + k+7$ number of colors, $O(\log (\Delta + k))$ bits of persistent storage in $O(m^{2})$ epochs, where $m$, $\Delta$, and $k$ are the number of robots, maximum degree of the graph and number of Doors, respectively.
\end{itemize}

\subsection{Organization}
Section~\ref{sec:model_app} describes the model used in this paper.
In Section~\ref{sec:singledoor_app}, we consider the MIS filling problem with a single Door, i.e., $k = 1$ for robots under ASYNC scheduler and in Section~\ref{sec:multidoor_app} with multiple Doors, i.e., $k > 1$ for robots under SSYNC scheduler before concluding in Section~\ref{sec:conclusion_app}.

\section{Model}\label{sec:model_app}
In this paper, we model the environment as a graph. We say that the graph contains a set of vertices that are connected to Doors from where robots can enter. The number of Doors in the graph is not known to the robots, but the robots are equipped with colors to distinguish themselves if they have entered the graph from different Doors. We assume that there are a maximum of $k$ Doors attached to the graph.

\subsubsection{Graph:}
We consider an anonymous graph, i.e., the nodes of the graph are indistinguishable from each other. Each vertex $v$ of the graph contains port numbers corresponding to the incident edges from $[1,2, \ldots, \delta(v)]$, where $\delta(v)$ is the degree of the vertex $v$. 
Given an anonymous connected port labeled graph $H = (V', E')$, we construct a graph $G = (V,E)$ with $k$ Doors, that adds two auxiliary vertices $\{d_i, d_i'\}$ corresponding to each Door that is connected to distinct vertices $\{v_1, v_2, \ldots, v_k\} \subset V'$.
We have a path $d_i \rightarrow d_i' \rightarrow v_i$ corresponding to each Door $d_i$.
The robots enter the graph through the Door, and a new robot appears at the Door immediately after it becomes empty.
We say a vertex is \emph{free} if none of the vertices adjacent to it are occupied by any robot.
Since we add a buffer vertex corresponding to each Door, all the vertices in the original graph $H$ are \textit{free} vertices in the beginning.

\subsubsection{Robots:}
The focus is on completing the task using robots with minimal capabilities operating under certain adversarial conditions.
The robots are \emph{autonomous} (no central or external control), \emph{myopic} (they have limited visibility range), \emph{anonymous} (without distinguishable features or identification) and \emph{homogeneous} (they all have the same capabilities/execute the same program).
Additionally the robots are \emph{luminous}, i.e.,
they have a light attached to them which can display various \emph{colors} that represent the value of a state variable. This works as a mode of communication between the robots.

\subsubsection{Time Cycle:}
Each robot operates in the \emph{Look-Compute-Move} (LCM) model, in which the actions of the robots are divided into three phases.
\begin{itemize}
    \item \textit{Look:} The robot takes a snapshot of its surroundings, i.e.,
          the vertices with in the visibility range and the colors of the robots occupying them.
    \item \textit{Compute:} The robot runs the algorithm using the snapshot as the input and determines a target vertex or chooses to remain in place.
    \item \textit{Move:} The robot moves to the target vertex if needed. A robot moves two hops in a single move phase. 
\end{itemize}

\subsubsection{Assumptions:}
We have the following assumptions regarding the knowledge of a robot and the properties of the underlying graph.
\begin{itemize}
    \item The robots have no knowledge of the graph but an upper bound of $\Delta$, the maximum degree of the graph.
    \item For a robot placed at $v$ with a visibility range of $z$, the port numbers of all the vertices in its visibility range are visible.
    \item A robot knows all the colors, but they can only display one color corresponding to the Door via which it enters the graph. Note that, this color corresponding to the Door only used by one robot at a time, but all robots that come from the same Door can display it when the need arises.(We use unique colors for each Door to determine a hierarchy among them.)
    \item The movement of the robot is non-instantaneous.
\end{itemize}

Note that, we use directions and port-numbers interchangeably throughout the paper. Each port number corresponds to a DIR color, also the direction towards Successor or Predecessor.

\subsubsection{Problem:}
We define the MIS filling problem formally as follows:

\begin{problem}{(MIS filling problem)}
    Given an anonymous connected port labeled graph $G = (V, E)$ with $k$ Doors, the objective is to relocate robots that appear at Doors such that at termination, the robots occupy a set of vertices $V_1$ ($V_1 \subset V$) that forms a maximal independent set of $G$.
\end{problem}

\section{Algorithm for MIS Filling with Single Door}\label{sec:singledoor_app}
We now describe the algorithm called IND, which is inspired by the PACK algorithm \cite{hideg2020asynchronous} and uses the concept of Virtual Chain Method  \cite{hideg2018filling}. The robots move throughout the graph is similar to the depth-first search (DFS). We assume that the robots operate under an asynchronous (ASYNC) scheduler.
An epoch is the shortest time in which each robot not in the Finished state is activated at least once and performs an LCM cycle.
Each robot requires a visibility range of 3 hops, $O(\log \Delta)$ bits of persistent memory, and $\Delta+8$ colors.
\subsection{Preliminaries}
\subsubsection{Colors:} 
The colors used by the robots are described here.
\begin{itemize}
    \item ON - Used initially when a robot arrives at the door.
    \item DIR - $\Delta$ colors corresponding to a port number in $[1, \Delta]$.
    
    \item CONF - Used to confirm that first DIR color has been seen and received.
    \item CONFC - Used to confirm that CONF color has been seen and received.
    \item CONF2 - Used to confirm that second DIR color has been seen and received.
    \item CONF3 - Used to confirm that the Packed state is achieved.
    \item WAIT - Used by a Leader while waiting for the Packed State.
    \item MOV - Used when a robot is in movement.
    \item OFF - Used by a robot in the Finished state.
    
\end{itemize}
Note that, a DIR color pointing towards a successor is used as a special color to indicate change of leadership.

\begin{definition}($k$ hops Neighborhood of a vertex $v$) 
For a vertex $v$, we define $k$-hops Neighborhood of $v$ to be the set of all vertices that are $k$ hops away from $v$ and denote it  by $N_v^k$.
\end{definition}
\begin{definition}
($k$ hops Visibility Set of a robot $r$)
For a robot $r$ placed at a vertex of the graph, we define the $k$ hops Visibility Set of $r$ to be the set of all vertices which are within $k$ hops from the current location of $r$. We denote it by $V_r^k$.
\end{definition}

\subsection{IND Algorithm}

In this section, we present the rules for the robots that they follow to successfully form a chain with robots occupying alternative vertices. We say a vertex is \textit{free} if none of its neighbors contain a robot. The first robot that enters the graph is called the Leader robot.
We define \textit{Packed state} as the state of a chain where all the alternative vertices are occupied by robots. We define it formally as follows:

\begin{definition}
    {(Packed state)}
    \label{packedstate_def_app}
    Let $L$ be a positive odd integer and $P = \{v_1, v_2, \ldots,$ $ v_L\}$ be a path starting from the Door at $v_1$ and the leader at $v_L$
    such that every second vertex of $P$, i.e., $v_1, v_3, \ldots, v_{L-2}$ was visited by the Leader. A chain of robots is in a Packed state if the vertices $v_1, v_3, \ldots, v_{L-2}$ are occupied by follower robots.
\end{definition}
\begin{flushleft}
    
    \begin{table}[H]
    \caption{List of Variables used in algorithms IND and MULTIND.}
        \label{tab2_app}
        \centering
        \begin{tabular}{| p{0.2\textwidth} | p{9.5cm} |}
            \hline
            \textbf{Variable} & \textbf{Description}                                                                     \\
            \hline
            State             & State of the robot - None/Follower/Leader/Finished                                       \\
            \hline
            Target            & Directions to the target vertex                                                          \\
            \hline
            NextTarget        & Directions to the vertex to which the robot has to move after reaching the Target vertex (Used by Followers) \\
            \hline
            Successor         & Successor robot                                                                          \\
            \hline
            Predecessor       & Predecessor robot                                                                        \\
            \hline
            Color             & Color displayed by robot's light                                                         \\
            \hline
            Entry             & Previous location of the robot/location of the follower stored in terms of directions    \\
            \hline
        \end{tabular}
        \vspace{\baselineskip}

    \end{table}
\end{flushleft} \vspace{-1.8cm}

 Only the Leader is allowed to move in the Packed state.
WAIT color is used by the Leader as soon as it reaches its target vertex to indicate that it is waiting for the chain to reach the Packed state. This is so that the Leader can choose a target such that neither the target nor any of the vertices adjacent to it are going to be occupied by any other robot. Secondly, the Leader (Predecessor) $r_1 $ needs to communicate the directions it will take to its follower (Successor) $r_2$ so that $r_2$ can know in which direction $r_1$ has moved. This applies to all predecessors and successors. The $Color$ variable represents the color displayed by the robot's light, and the $Target$ variable represents directions to the target vertex. The $NextTarget$ variable represents directions to the next target after reaching the target vertex. $Entry$ variable represents the direction of the two hops a robot moved, so that the robot knows the location of its follower.

  We explain communication between  Predecessor and Successor  and the restoration of Packed state after the movement with an example as shown in Fig. \ref{fig:color_app}.

\subsubsection{Communicating the movement directions:}\label{sec:communicatingmovement_app} The robots establish the Predecessor and Successor relationship among them by their order of appearance from the Door. A Predecessor $r$ communicates its destination $r.Target$ to a Successor by showing the colors corresponding to the port numbers at the vertex. Suppose $r_1$, $r_2$ and $r_3$ are located at $e$, $c$ and $a$, respectively as shown in Fig.~\ref{fig:color_app}.

\begin{figure}[H]

    \centering
    \includegraphics{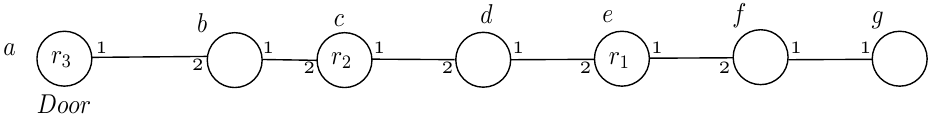}
    \caption{Communication of color from $r_1$ to $r_2$}
    \label{fig:color_app}
    
\end{figure}

$r_1$ determines that it will move to vertex $g$. Now, it has to communicate that to $r_2$, so that $r_2$ can follow $r_1$. First $r_1$ shows the DIR color of port 1 corresponding to the edge $ef$ as it wants to take the path $ef-fg$, and $r_2$ responds by showing CONF color. Then $r_1$ sets its color to CONFC to confirm that it has seen the CONF color. Now, $r_2$ also sets its color to CONFC to inform $r_1$ that it is ready to receive the second color. (CONFC color is used to distinguish between two DIR colors that can be the same.) Next, $r_1$ shows the DIR color of port 1 corresponding to edge $fg$. $r_2$ confirms that it has seen the color by setting its color to CONF2. Once $r_1$ sees CONF2, it can move to $g$ after setting its own color to MOV. Now $r_2$ sets $e$ as its target (which was the old position of $r_1$ 
Then, $r_2$ does the same process to communicate DIR colors to $r_3$.

\subsubsection{Restoration of Packed state after movement:}
We describe this module with the Fig. \ref{fig:color_app}. After the leader $r_1$ moves to $G$ with color MOV, $E$ becomes empty, so the Packed state is now distorted. So $r_1$ changes its color to WAIT. Now, $r_2$ moves to $E$ with color MOV after communicating DIR colors to $r_3$. So, $C$ becomes empty. After that, $r_3$ moves to $C$ with color MOV without communicating DIR colors as it does not have any successor at this movement. As soon as $r_3$ leaves the Door, a robot $r_4$ is placed at the Door with color ON. After $r_3$ reaches $C$, it sees $r_4$ with color ON and $r_3.Color=$ MOV. So, it changes its color to CONF3. Now, $r_2$ also changes its color to CONF3 after seeing $r_3.Color =$ CONF3 and $r_2.Color =$ MOV. The leader $r_1$ sees $r_2$ with color CONF3 and $r_1.Color=$ WAIT, so it understands that the Packed state is achieved. Now, $r_1$ looks for new target to move to. 

\subsubsection{Transferring the leadership:} 
The current leader $r_1$ transfers its leadership if either of the two scenarios occurs.
First, $r_1$ gets stuck i.e., there are no other free vertices left to move to from the current vertex of the leader.
 Secondly, for each free vertex $v\in V_{r_1}^2$, $r_1$ finds at least one vertex $v' \in N_v^2 \cap V_{r_1}^3$ with a robot not in Finished state.
In both of the above cases, $r_1$ transfer the leadership to its successor $r_2$ by communicating the direction that points towards $r_2$.

\subsubsection{Detailed description:}

When a robot first appears at the Door, it initializes to the None state and sets color ON.
Let $r_1$ be the first robot that appears at the Door. $r_1$ does not find a robot on any adjacent vertices, so it changes its state to Leader. Now $r_1$ is the Leader and chooses a target vertex two hops away and moves to it with color MOV. As soon as the Leader $r_1$ leaves the Door, the next robot $r_2$ appears at the Door. At this time, $r_1$ is still in motion and is nearest to $r_2$. $r_2$ sees this and becomes the follower of $r_1$ and sets $r_2.Color=$ ON, where $r_1.Successor = r_2$ and $r_2.Predecessor = r_1$. After $r_1$ reaches its target, it sets $r_1.Color=$ WAIT to indicate that it is waiting for Packed state. When $r_2$ sees $r_1.Color=$ WAIT and $r_2.Color=$ ON, $r_2$ changes its color to CONF3. $r_1$ now chooses a new free vertex (if any) as $r_1.Target$ and communicate its directions to $r_2$ as described above. When $r_1$ gets confirmation from $r_2$ ($r_2$ sets its color to CONF2), $r_1$ moves to $r_1.Target$ with color MOV and again changes its color to WAIT to indicate that it is waiting for the chain to be in Packed state.

When the chain is in Packed state, the leader $r_1$ chooses a free vertex $v$ from $V_{r_1}^2$ as target, if all $v' \in N_v^2 \cap V_{r_1}^3$ are either unoccupied or having robots at Finished state. Then,
$r_1$ communicates the directions of $r_1.Target$ to $r_2$ as follows: $r_1$ displays DIR color corresponding to $r_1.Target.One$ by setting $r_1.Color$. $r_2$ sees this and stores the direction in $r_2.NextTarget.One$. $r_2$ confirms that it has seen the first DIR color by displaying CONF color. $r_1$ confirms that it has seen CONF on $r_2$ by setting $r_1.Color$ to CONFC (confirmation of confirmation). $r_2$ sees CONFC on $r_1$ and in turn, sets $r_2.Color$ to CONFC to show that it is ready to receive the second DIR color. The second DIR color $r_1.Target.Two$ is displayed by $r_1$ and is seen by $r_2$. $r_2$ stores this second direction in $r_2.NextTarget.Two$ and sets $r_2.Color$ to CONF2 to send confirmation that the second direction of $r_1$ is received.
Once $r_1$ has seen CONF2 in $r_2$, it moves two hops to $r_1.Target$ in Move phase by setting $r_1.Color=$ MOV. After reaching to $r_1.Target$, $r_1$ sets $r_1.Color=$WAIT. Now $r_2$ sets $r_2.Target$ based on the information stored in $r_2.NextTarget$. $r_2$ needs to reach the old position of $r_1$.

Finally, when the Leader $r_1$ can no longer find free vertices to move to, it communicates this information to its follower $r_2$ using the DIR color that points towards $r_2$. Then, $r_1$ changes its color to OFF and goes into the Finished state, and $r_2$ becomes the Leader and continues exploring the graph.

 When a robot moves to its target vertex, it sets $r.Entry$ to the directions of the two hops it moved under \emph{Entry} variable (so that the robot knows the location of its follower).

\subsubsection{Pseudocode of IND Algorithm:} 
The subroutine \textsc{Communication($r.Target$)} communicates the target set by a robot $r$ to its successor. 
The subroutine \textsc{Receive($r.Predecessor.Target$)} works when $r$ receives directions from its predecessor. \textsc{Packed\_State($r,r.Successor$)} works for achieving packed state after the movement of the robots. \textsc{Leadership\_Transfer($r$)} works when the current leader $r$ finds no vertex to move to.

\begin{algorithm2e}
	
	\If{$r.Target$ is set \& $r.Successor.Color=$CONF3}
	{
		Set $r.Color$ to match $r.Target.One$ \algorithmiccomment{(Showing the first DIR color)}\\
	}
	\ElseIf{$r.Target$ is set \& $r.Successor.Color=$ CONF}
	{ Set $r.Color=$ CONFC \algorithmiccomment{(Confirming the confirmation sent by successor)}\\
	}
	\ElseIf{$r.Target$ is set \& $r.Successor.Color=$ CONFC}
	{ Set $r.Color$ to match $r.Target.Two$ \algorithmiccomment{(Showing the second DIR color)}
	}
	\ElseIf{$r.Target$ is set \& $r.Successor.Color=$CONF2}
	{Set $r.Color=$ MOV 
		\algorithmiccomment{(Moving towards the target)}\\
		Move to $r.Target$\\
	}
	\caption{\textsc{Communication}($r.Target$)}
\end{algorithm2e}

\begin{algorithm2e} 
        \If{$r.NextTarget$ is not set}
        {
            \If{$r.Color=$ CONF3 \& $r.Predecessor$ shows a DIR color}
            {
                Store that shown color as $r.NextTarget.One$\\
                Set $r.Color=$ CONF
                \algorithmiccomment{(Confirmation for the first DIR color)}\\
            }
            \ElseIf{$r.Color=$ CONF \& $r.Predecessor.Color=$ CONFC}
            {
                Set $r.Color=$ CONFC
                \algorithmiccomment{(Ready to receive the second DIR color)}\\
            }
            \ElseIf{$r.Color=$  CONFC \& $r.Predecessor$ shows a DIR color}
            {
                 Store that shown color as $r.NextTarget.Two$\\
                Set $r.Color=$ CONF2
                \algorithmiccomment{(Confirmation for second DIR color)}\\
       }
   }
   \caption{\textsc{Receive}($r.Predecessor.Target$)}
   \label{receive_app}
\end{algorithm2e}

\begin{algorithm2e} 
        \If{$r.Color=$ MOV \& $r.Successor.Color=$ ON}
        {
             Set $r.Color=$ CONF3 \algorithmiccomment{(For $r.Successor$ being  at the Door)}\\
            }
        \ElseIf{$r.Color=$  MOV \& $r.Successor.Color=$CONF3}
        {
            Set $r.Color=$  CONF3 \algorithmiccomment{(For any other pair of predecessor and successor)}\\
        }
        \ElseIf{$r.Color=$ WAIT \& $r.Successor.Color=$ CONF3}
        {
             return \algorithmiccomment{(For the leader and its successor)}\\
            }

       \caption{\textsc{Packed\_State}($r$, $r.Successor$)}
    \label{packedstate_app}
\end{algorithm2e}

\begin{algorithm2e}
	\If{$r.Successor$ is set}
	{
	$r$ sets DIR color to point towards its Successor\\
	$r.Color=$OFF \\
	Change $r.State$ to Finished \\
}
\Else
{	$r.Color=$OFF \\
	Change $r.State$ to Finished \\
	}
\caption{\textsc{Leadership\_Transfer}($r$)}
\label{leadershiptransfer_app}
\end{algorithm2e}

\begin{algorithm2e}
	\If {$r.State$ is Leader}
		{
				\If {$r.Target$ is not set}
				{
						\If {$r.Entry$ is set and $r.Successor.Color=$ CONF3}
						{
							\If{$\exists$ a vertex $v \in  V_r^2$ such that all $v' \in V_r^3 \cap N_v^2$ is either unoccupied or occupied with robots in Finished state}
							{
								$r$ sets $v$ as $r.Target$ by setting $r.Target.One$ and $r.Target.Two$ \\ 
								\textsc{Communicate}($r.Target$)  \\
								Set $r.Color=$ WAIT\\
								$r$ clears $r.Target$ by clearing $r.Target.One$ and $r.Target.Two$\\
								\textsc{Packed\_State}($r$, $r.Successor$)\\
							}
							\Else 
							{
								\textsc{Leadership\_Transfer}($r$)
							}
						}
					}
				}
			
\If{$r.State$ is Follower}
		{
			\If {$r.NextTarget.One$ is not set}
			{
				 	Receive($r.Predecessor.Target$)\\
				 }
			
			\Else
			{
				\If{$r.Target$ is set}
				{
						\textsc{Communicate}($r.Target$) \\
						Clear $r.NextTarget.One$ and $r.NextTarget.Two$\\
						\textsc{Packed\_State}($r$, $r.Successor$)\\
					}
				}
			}

\If{$r.State$ is None \& $r.Color=$ON}
{
			\If {$r$ does not find any robot within a distance of 2 hops}
			{
					Find a vertex $v \in V_r^2$\\
					Set $v$ as $r.Target$ by setting $r.Target.One$ and $r.Target.Two$\\
					Change $r.State$ to Leader\\
					Set $r.Color=$ MOV \\
					$r$ moves to $r.Target$ \algorithmiccomment{($r$ moves to $r.Target.Two$)}\\
					Set the direction of the 2 hops $r$ moves as $r.Entry$\\
					$r$ sets $r.Color$ to WAIT\\
				}
			\Else
			{
					$r$ chooses the nearest robot as $r.Predecessor$\\
					Change $r.State$ to Follower\\
					\If{$r.Predecessor.Color=$ WAIT \& $r.Color=$ON}
					{
						Set $r.Color=$ CONF3\\
					}
			}
		}

\caption{IND}
\label{ind_app}
\end{algorithm2e}

\subsubsection{Example} \label{exampleIND_app}
Now, we illustrate the execution procedure of Algorithm (IND) with the following example.
The location of the Door and the initial setting is shown in Fig.~\ref{fig:indflow1_app}. The robot $r_{1}$ first appears at the Door in Fig.~\ref{fig:indflow2_app}. As soon as $r_{1}$ moves away from the Door, robot $r_{2}$ appears there and sets $r_{1}$ as its predecessor. This situation is shown in Fig.~\ref{fig:indflow3_app}. In Fig.~\ref{fig:indflow4_app}, $r_{1}$ reaches its target vertex but $r_{2}$ does not move until $r_{1}$ has moved away from its current position. $r_{1}$ reaches a new vertex and $r_{2}$ follows it, $r_{3}$ appears at the Door and follows $r_{2}$ which is shown in Fig.~\ref{fig:indflow5_app}. Observe that $r_{1}$ is stuck so it transfers its leadership to its follower, i.e., $r_{2}$. In Fig.~\ref{fig:indflow6_app}, $r_{2}$ reaches a new target vertex and its followers follow. 
\begin{figure}
    \label{indfig_app}
    \centering
    \begin{subfigure}{.45\textwidth}
        \centering
        \includegraphics[scale=0.15]{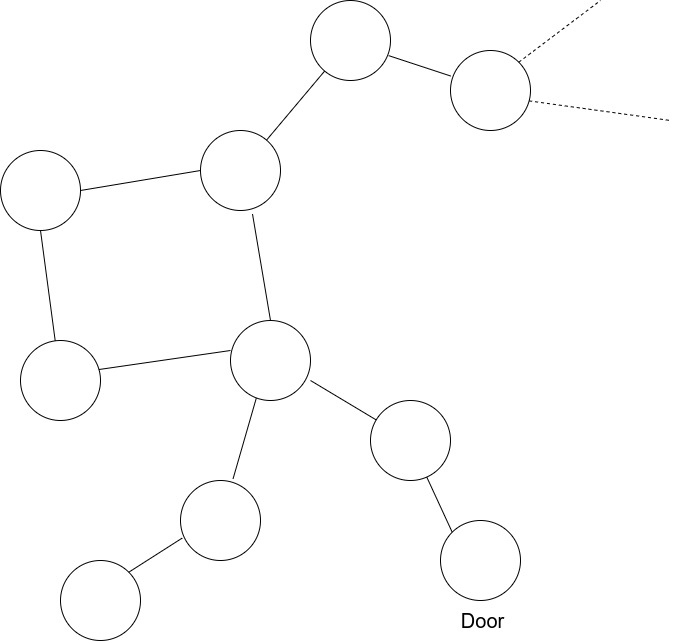}
        \caption{Location of the Door and \\ initial setting of the graph.}
        \label{fig:indflow1_app}
    \end{subfigure}\hfill
    \begin{subfigure}{.450\textwidth}
        \centering
        \includegraphics[scale=0.15]{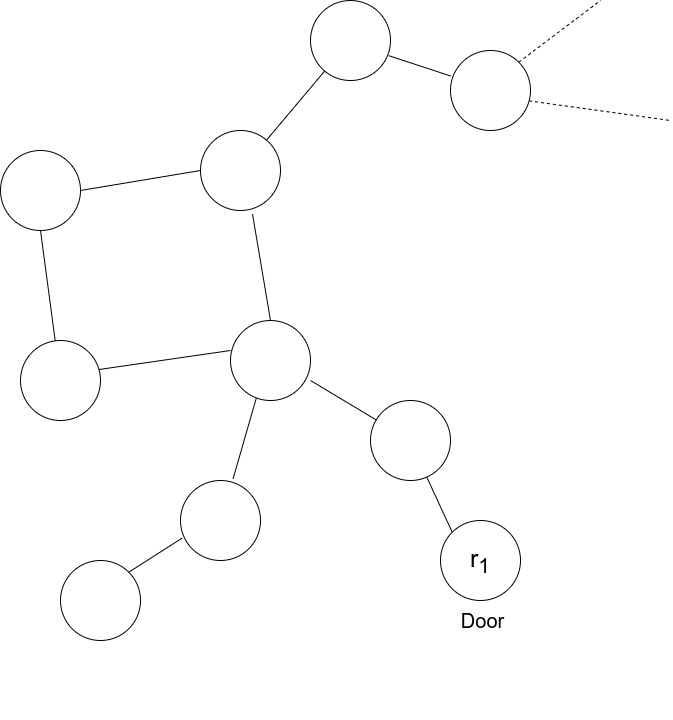}
        \caption{The robot $r_1$ appears at the Door \\ and becomes the Leader.}
        \label{fig:indflow2_app}
    \end{subfigure} 

    \begin{subfigure}{.45\textwidth}
        \centering
        \includegraphics[scale=0.15]{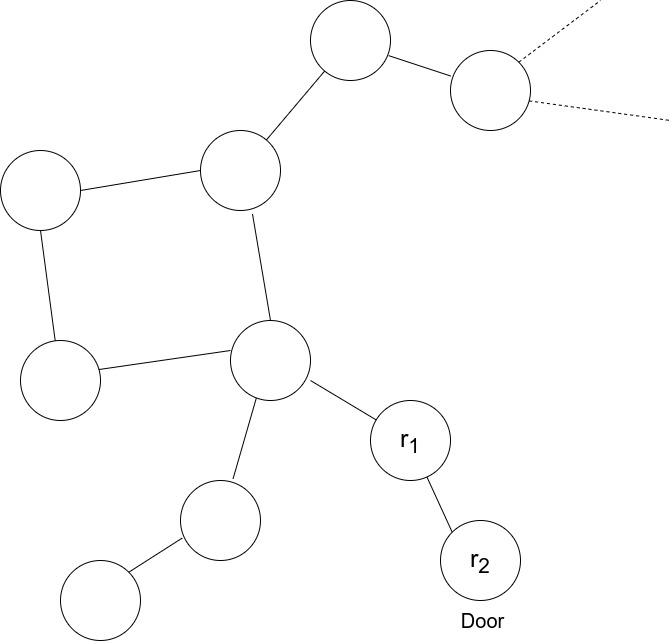}
        \caption{As soon as $r_1$ moves away from \\ the Door and $r_2$ becomes its follower.}
        \label{fig:indflow3_app}
    \end{subfigure}\hfill
    \begin{subfigure}{.450\textwidth}
        \centering
        \includegraphics[scale=0.15]{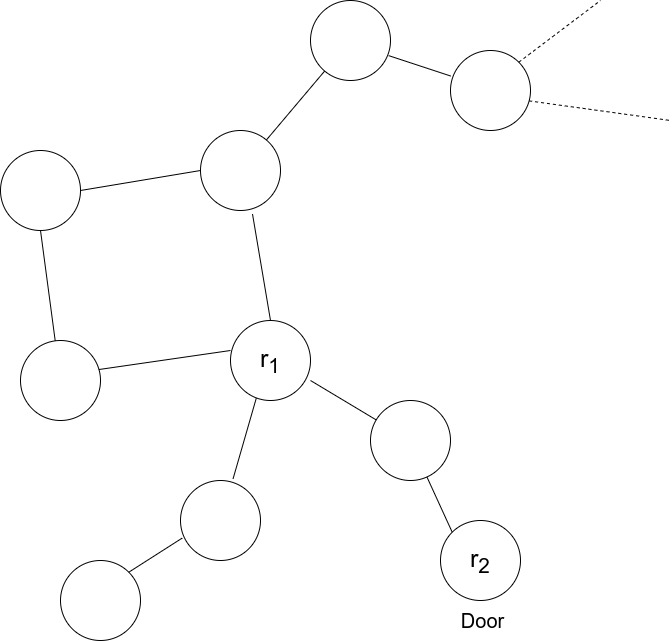}
        \caption{$r_1$ reaches its target vertex.}
        \label{fig:indflow4_app}
    \end{subfigure}
    \begin{subfigure}{.45\textwidth}
        \centering
        \includegraphics[scale=0.15]{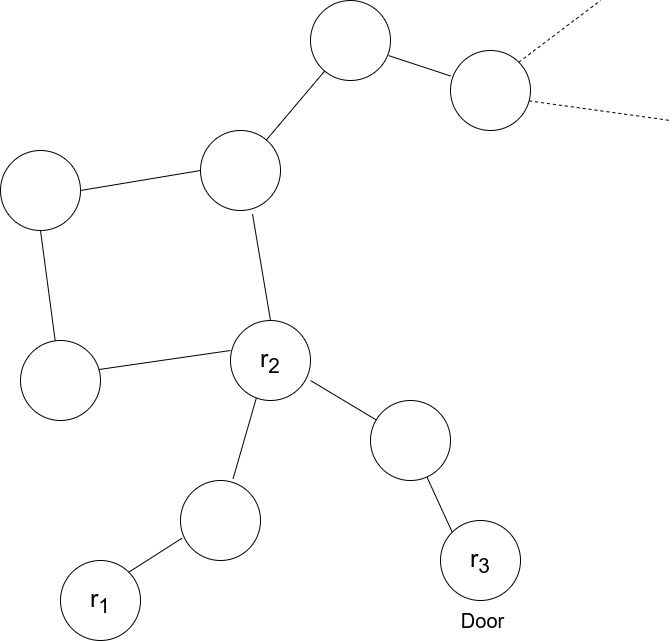}
        \caption{$r_1$ transfers its leadership to $r_2$ \\ as there are no free vertex for it.}
        \label{fig:indflow5_app}
    \end{subfigure}\hfill
    \begin{subfigure}{.45\textwidth}
        \centering
        \includegraphics[scale=0.15]{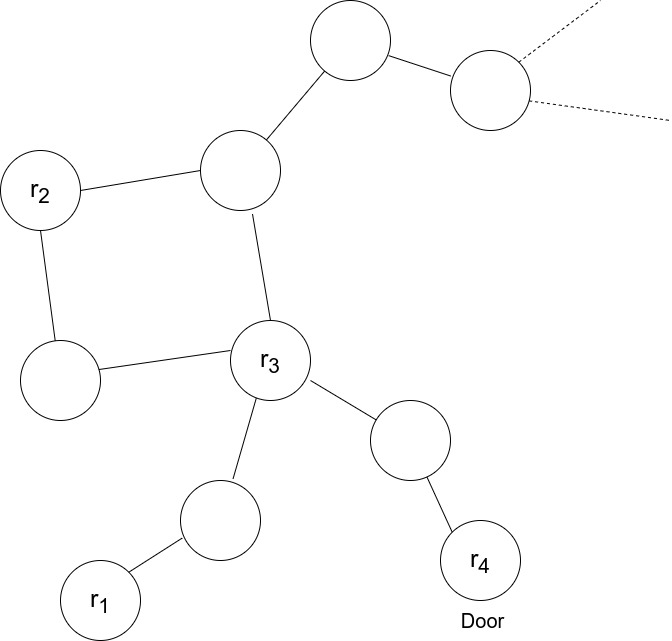}
        \caption{$r_2$ is now the Leader and moves to new free vertices.}
        \label{fig:indflow6_app}
    \end{subfigure}
    \begin{subfigure}{.5\textwidth}
        \centering
        \includegraphics[scale=0.15]{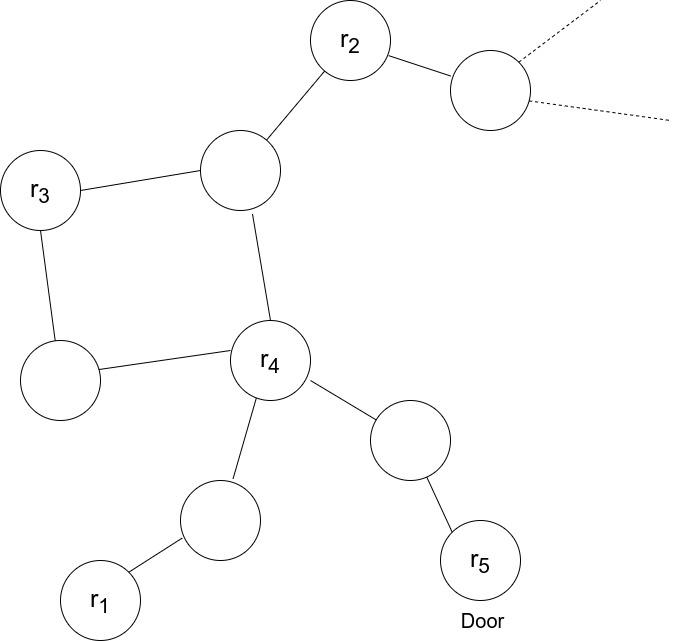}
        \caption{$r_2$ keeps exploring the graph.}
        \label{fig:indflow7_app}
    \end{subfigure}
    \caption{An example execution of the IND algorithm}
\end{figure}

\subsection{Analysis of IND Algorithm}
First, we present few lemmas that establish the behavior of the robots. The proofs of the following lemmas and theorem can be found in Appendix.

\begin{lemma}\label{lem1_app}
    There can be at most one Leader robot, and  the Leader robot $r_1$ moves to a free vertex $v \in V_{r_1}^2$ such that every $v' \in N_v^2 \cap V_r^3$ is either unoccupied or occupied by robots in Finished state.
\end{lemma}
\begin{proof}\label{lem1proof_app}
    In the rules for \emph{Transferring the Leadership}, when a Leader $r_1$ signals to its successor $r_2$ that it is stuck, $r_2$ can become the Leader only after the previous Leader $r_1$ has switched to Finished state (recognized by OFF color on $r_1$).
    The first robot placed becomes the Leader, and the robots appearing next can become a Leader only after the previous one reaches the Finished state.
    Therefore, there can be at most one Leader at any time during the dispersion.

    A free vertex has none of its adjacent vertices occupied by a robot.
    As the current Leader $r_1$ can only move when the chain is in a Packed state, it chooses a vertex in $V_{r_1}^2$ that is free by checking all the vertices adjacent to potential target vertex are not occupied by any robot.
    This can be done since the robots have a visibility range of 3.
    Further, there is no possibility of this target vertex or any vertex adjacent to it being occupied by any other robot as the Leader is allowed to move only when the chain is in the Packed state. If there is a free vertex $v \in V_{r_1}^2$ as a potential target, then every $2$ hops neighbor $v'$ of $v$ ($v' \in N_v^2 \cap V_{r_1}^3$), that is visible to $r_1$ needs to be either free or occupied by robots in Finished state. Due to this condition, a chain does not cross itself as shown in Fig. \ref{selfcross_app}. If no such free vertex $v$ is found, then $r_1$ transfers the leadership to $r_2$ and switches to Finished state. \qed
\end{proof}

\begin{corollary}
\label{corollary1_app}
    The leader $r$ moves to a free vertex $v$ such that every vertex $v' \in N_v^2 \cap V_{r}^3$ is either unoccupied or occupied by a robot in Finished state. Consequently, the chain never crosses itself.
\end{corollary}

\begin{lemma}\label{lem2_app}
    The Robots do not collide.
\end{lemma}
\begin{proof}\label{lem2proof_app}
     When a robot appears at the Door, it has only one neighbor robot except for the first robot.
    That neighbor robot will become the Predecessor, which the robot will follow during the movement.
    The predecessor shows the directions to the Successor in which it will move,   before moving so that the Successor can always follow it.
    After the target is set, the Predecessor moves first and only then the follower moves in the direction of its Predecessor.
    Each Successor robot has one Predecessor, which means they cannot collide with each other. Also, from Corollary \ref{corollary1_app}, a chain does not cross itself and hence the collision never happens.  \qed  
    
\end{proof}

\begin{lemma}\label{lem3_app}
    No two robots in the Finished state occupy adjacent vertices.
\end{lemma}

\begin{proof}\label{lem3proof_app}
    A robot can go to the Finished state only from the Leader state.
    From Lemma~\ref{lem1_app}, we know that a Leader robot moves only to free vertices.
    As only free vertices are occupied, the robots are not present in adjacent nodes. \qed
\end{proof}

From Lemma~\ref{lem3_app}, we can say that when a robot enters the Finished state, none of the vertices adjacent to it are occupied by a robot which means eventually, all the vertices occupied by robots form an independent set. 

\begin{remark}
If a vertex is occupied by a robot in the packed state, it remains occupied thenceforth.
\end{remark}

\begin{theorem}\label{thm1_app}
    Algorithm IND fills a maximal independent set.
\end{theorem}

\begin{proof}
    From Lemma~\ref{lem3_app}, the filled vertices form an independent set.
    Suppose the independent set is not maximal.
    As the graph is connected, there exists a vertex $v$ which is free and has a robot $r_1$ in the Finished state two hops away since the execution of the algorithms is done.
    Before the robot $r_1$ entered the Finished state, it was in the Leader state.
    The Leader $r_1$ switches into the Finished state in two cases. First, it cannot find any free vertex two hops away. Secondly, it finds a free vertex $v$ such that there exists a neighbor $v'$ of $v$ in $N_v^2 \cap V_{r_1}^3$ occupied by a robot not in Finished state.
    If all the neighbors of $v$ in $N_v^2 \cap V_{r_1}^3$ are unoccupied or have robots in Finished state, $r_1$ has a clear path to $v$. Hence, it cannot switch to Finished state, which leads to a contradiction. If the 2 hops neighbor $v'$ of $v$ has a robot not in Finished state, $r_1$ does not have a clear path to $v$. 
    In that case, $r_1$ transfers the leadership to its successor by communicating the direction pointing towards its successor and eventually the robot on $v'$ will becomes the leader. Since $v'$ is 2 hops away from $v$, the robot on $v'$ cannot switch to Finished state as it gets a clear path to the free vertex $v$, which again leads to contradiction.
    
    \qed
\end{proof}

\begin{lemma}
    Algorithm IND fills an MIS of $G$ with luminous robots having visibility range of 3 hops, $O(\log \Delta)$ bits of memory, and $\Delta+8$   colors.
\end{lemma}
\begin{proof}
    
    The visibility range of 3 hops for robots executing IND algorithms is necessary. Otherwise, the robots cannot check if the target vertex is free or not.

    The robots require $O(\log \Delta)$ bits of memory to store the following: \emph{State} (4 states: 2 bits), \emph{Target} (directions to the target vertex: $2  \ceil{\log \Delta}$ bits \footnote{Two port numbers are needed to be stored corresponding to the movement.}), \emph{NextTarget} (directions to the vertex to which the robot has to move after the \emph{Target} vertex  is reached: $2  \ceil{\log \Delta}$ bits).

    The colors used by the robots are $\Delta$ colors to show the directions to the target of the robot, including the special color that points in the direction of Follower to switch to the Finished state. Initially when the robot is placed at the Door, the robot sets its color to ON.
    There are four additional colors (CONF, CONFC, CONF2 and CONF3) for confirming that the robot saw the signaled direction and confirmations of the Predecessor or the Successor and the Packed state is achieved, one color WAIT used when the Leader is waiting for the Packed State, one color MOV used during the movement and the OFF color.\qed
\end{proof}

Now we analyze the time complexity of the algorithm in terms of epochs.
To find the total time required by the algorithm, we first individually establish the time-bound on the movement of robots.
Consider a chain containing $i$ robots $\{r_1, r_2, \ldots, r_i\}$, where $r_1$ is the Leader and the foremost robot in the chain. The robots $r_1, r_2, \ldots, r_i$ are on alternative vertices on the path from the vertex occupied by the Leader to the Door vertex. Suppose the chain is in the Packed state. 
We determine the time required for two consecutive movements of the Leader.
We first determine the time required for all the robots in the chain to occupy the position of their Predecessor. As a result, a new robot appears at the Door, increasing the size of the chain. 
Next, we find the time required for the robots in the chain to set their colors to CONF2, indicating that they have received the movement direction of their Predecessor.
 We also find the time required for  the chain to reach in the Packed state again after the leader of the chain moves.
\begin{lemma}\label{t1lem_app}
    Algorithm IND takes at most $i$ epochs for all the robots in a chain of length $i$ to perform one MOVE operation.
\end{lemma}

\begin{proof}
    Once the chain is in the Packed state, the leader robot $r_1$ moves to its target vertex in one epoch and $r_1$ sets its color to WAIT. By the next epoch, $r_2$ observes that $r_1$ has left its previous vertex, so it moves to $v$. In a cascading manner all Successors move to their Predecessors location. Thus in the worst case, it takes $i$ epochs for all the robots in the chain to move.\qed
    
\end{proof}

\begin{lemma}\label{t2lem_app}
    Algorithm IND takes at most $7i$ epochs for the robots in a chain of length $i$ to reach the Packed state after movement. 
    
\end{lemma}

\begin{proof}
    As per the algorithm description, the communication between a Predecessor $r_p$ and its Successor $r_q$ by showing a series of seven colors: first DIR color at $r_p$, CONF at $r_q$, CONFC at $r_p$, CONFC at $r_q$, second DIR color at $r_p$, CONF2 at $r_q$ and CONF3 at $r_q$. A leader can only move when the chain is in the Packed state. For a chain to be in the Packed state, the Successor of the leader robot has to show CONF3 color.
    The process of communication of colors starts from the new robot $r_{i+1}$ that appears at the Door. It takes at most seven epochs for the communication between $r_i$ and $r_{i+1}$. 
    Similar communication happens between all $i$ pairs of consecutive robots on the chain. So it takes at most $7i$ epochs for the chain to reach Packed state.\qed
\end{proof}

\begin{lemma}\label{t3lem_app}
    Algorithm IND takes at most 4 epochs to have a transfer of leadership from a leader to its Successor.
\end{lemma}

\begin{proof}
    Now consider the situation where the Leader $r_1$ cannot find any free vertices two hops away.
    If $r_1$ is at the Door, it sets its color to OFF and switches to Finished state; the maximal independent set is filled.
    Otherwise, $r_1$ switches to the Finished state by setting the DIR towards its follower.
    $r_2$ sees the DIR color and sets $r_2.Color$ to CONF.
    $r_1$ sees the CONF color, it switches to color OFF.
    $r_2$ becomes the new Leader once it sees the color OFF at $r_1$. In total, this needs at most 4 epochs because of the sequence of colors (DIR, CONF and OFF) and final state change at $r_2$.\qed
\end{proof}
\begin{theorem}\label{thm:indtime_app}
    The algorithm IND runs in $O(m^{2})$ epochs.
\end{theorem}

\begin{proof}

    For an MIS of size $n$, we can have a chain of length at most $n$. 
    For each increase in a chain of length $i$, it takes at most $7i$ epochs from Lemma~\ref{t1lem_app} and \ref{t2lem_app}.
    Also, we can have at most $m$ leadership transfers.
    From Lemma~\ref{t3lem_app}, each transfer of leadership takes at most 4 epochs. 
    So in total we need, $\sum_{i=1}^m 7i + 4 m = O(m^2)$ epochs.
    Therefore, after at most $O(m^{2})$ epochs, an MIS of vertices of the graph becomes filled.\qed
\end{proof}

The corollary below follows from Corollary 1 in \cite{hideg2020asynchronous}.
\begin{corollary}
    (i) Assume that there are no inactive intervals between the LCM cycles and that every LCM cycle of every robot takes at most $t_{max}$ time.
    Then the running time of the IND algorithm is $O(m^{2}t_{max})$. (ii) The IND algorithm needs $O(m^{2})$ LCM cycles under an FSYNC scheduler.
\end{corollary}

\section{Algorithm for MIS Filling with Multiple Doors}\label{sec:multidoor_app}

The MULTIND is largely similar to algorithm IND with a few modifications.
It works under a \emph{Semi Synchronous} (SSYNC) scheduler, where a subset of robots is activated in each round, where each activated robot finishes its LCM cycle in the same round.  We define an epoch similarly as the minimum number of rounds where all robots are activated at least once. Note that an epoch can have a variable number of rounds.
The graph has a maximum of \emph{k} number of Doors.
The robots do not have any knowledge about the number of Doors. 
The visibility range of the robots is five hops.
Each robot has $O(\log (\Delta + k))$ bits of persistent memory, and $\Delta + k + 7$ colors where $\Delta$ number of DIR colors, CONF, CONFC, CONF2, CONF3, MOV, ON, OFF and $k$ number of WAIT colors denoted by WAIT-1, WAIT-2, $\ldots$, WAIT-$k$ representing that a robot is waiting as well as their rank.
All the robots entering from a particular Door can only display WAIT color corresponding to that Door.
The WAIT colors can be compared against each other to establish dominance between the robots. Initially a robot has color ON when it is placed at the Door.
The proposed MULTIND runs in O($m^{2}$) epochs.

\subsection{The MULTIND Algorithm}
The Leader robots need to avoid collision with other Leader robots and the follower robots in chains.
The robots make use of the hierarchy among the $k$ WAIT colors to avoid collision with another Leader robot.
The Leader robots also avoid cutting through a chain to avoid collision with follower robots in another chain.
In the multiple Door situation, the Leaders display the WAIT-$i$ color instead of the WAIT color used in a single Door case.
In the \emph{Look} phase, if a Leader robot $r_{i}$ with color WAIT-$i$ sees any other Leader $r_{j}$ with color WAIT-$j$ such that $j < i$, then we say that the Leader robot $r_{j}$ \emph{dominates} Leader robot $r_{i}$.
The Leader robot $r_{j}$ is said to be the \emph{dominating} and the Leader robot $r_{i}$ is said to be \emph{dominated}.
Note that a dominating Leader robot can be dominated by another Leader robot at the same time.
If a Leader robot is not dominated by any other robot, it can choose a target.

The rest of the model is the same as in the Single Door case.
The robots can be in any one these states during execution: \emph{None}, \emph{Leader}, \emph{Follower}, \emph{Finished}.
When the robots appear at the Door, they are initialized with None state and with color ON.
We define $\mathcal{P}_k(v_i, v_j)$ as the set of vertices that are part of all the paths from $v_i$ to $v_j$ of length $k$. Note that there can be multiple  paths of length $k$, and we include vertices of all those paths.
The additional rules are applicable to the Leader as follows. 

\subsubsection{Movement of a Leader robot:} If a chain is in Packed state, the corresponding leader $r_L$ chooses a free vertex in $V_{r_L}^2$ as target in one of the following way.
\begin{itemize}
    \item If $V_{r_L}^5$ has a dominating leader $r'_L$ of some other chain, then a free vertex $v \in V_{r_L}^2$ is chosen such that  $v \notin \mathcal{P}_5(r_L,r'_L)$. If no such vertex $v$ is found, then $r_L$ transfers the leader to its successor by pointing the direction towards the successor and goes into Finished state by changing its color to OFF.
    \item If $r_L$ is the leader dominating some other leader $r'_L$ in $V_{r_L}^5$, then a free vertex $v\in V_{r_L}^2$ is chosen such that every vertex $v' \in N_v^2 \cap V_{r_L}^2$ is either unoccupied or occupied with robot in Finished state. If no such vertex $v$ is present, $r_L$ transfer its leadership to its successor.
    \item If $V_{r_L}^5$ does not have any other leader robot, then a free vertex is chosen as target such that every vertex $v' \in N_v^2 \cap V_{r_L}^2$ is either unoccupied or occupied with robot in Finished state. If no such free vertex is there, $r_L$ transfers the leadership to its successor.
\end{itemize}

\begin{figure}
\centering
\begin{minipage}[b]{0.45\textwidth}
\centering
  \includegraphics[width=0.95\textwidth]{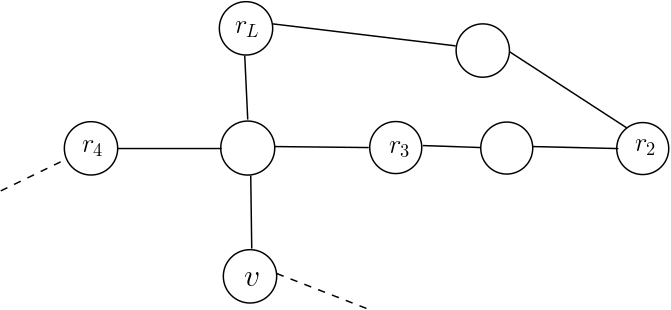}  
  \caption{A chain $r_L-r_2-r_3-r_4\cdots$ does not cross itself}
\label{selfcross_app}
\end{minipage}
\hspace{0.05cm}
\begin{minipage}[b]{0.45\textwidth}
\centering
  \includegraphics[width=0.95\textwidth]{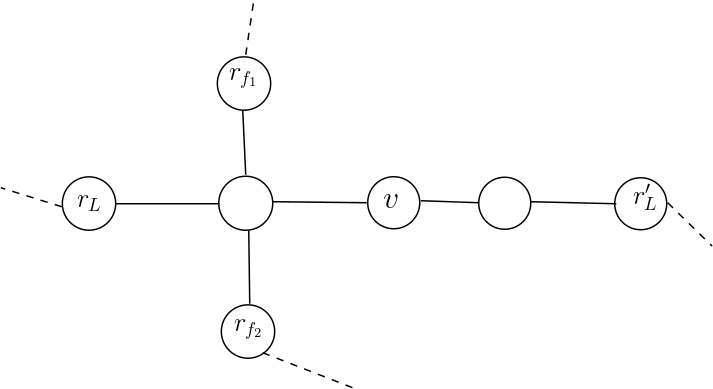}  
  \caption{Two chain do not cross each other }
  \label{crossing_app}
\end{minipage}

\end{figure}

After the target is fixed for a leader $r_L$, it communicates the target to its successor using DIR colors. After getting confirmation from its successor, it moves to the target with color MOV and waits for the chain to be in Packed state with respective WAIT color. Additionally, irrespective of whether a Leader robot is dominated or not, the target chosen is such that the Leader robot does not cut through a chain while moving to that target vertex, i.e., the vertex between the target vertex and the current vertex does not have more than one of its neighboring vertices occupied at least one of which is an active robot (a robot not displaying OFF color) other than the current Leader robot. As shown in Fig. \ref{crossing_app}, $r_L$ and $r'_L$ are the dominating and dominated Leader at an instance and $v$ is a free vertex which is 2 hops away from $r_L$. $r_{f_1}$ and $r_{f_2}$ are some followers of some other chain. $r_L$ does not reach to $v$ even if $r_L$ is the dominating Leader, as $r_{f_1}$ and $r_{f_2}$ is not in Finished state. As shown in Fig. \ref{selfcross_app}, a Leader $r_L$ sees a free vertex $v$ two hops away. But $r_L$ cannot reach to $v$ (self cross) as the two hops neighboring vertices are occupied by $r_3$ and $r_4$.

\subsubsection{Pseudocode of MULTIND:}
The pseudocode of the  MULTIND (Algorithm \ref{multind_app}) is given below. The subroutine \textsc{MULTIND\_FindTarget($r$)} is used for finding $r.Target$. All the  other subroutines such as  \textsc{Communication($r.Target$)}, \textsc{Receive($r.Predecessor.Target$)}, \textsc{Packed\_State($r,r.Successor$)} and \textsc{Leadership\_Transfer($r$)} remain same as IND algorithm.

\begin{algorithm2e}
	
	\If{$\exists$ a dominating leader $r' \in V_r^5$}
	{
		\If{$\exists$ a free vertex $v \in V_r^2$ such that $v \notin \mathcal{P}(r,r')$}
		{ 
			Set $v$ as $r.Target$
		}
		\Else
		{
			\textsc{Leadership\_Transfer}($r$) \algorithmiccomment{(No such free vertex to move to)}
		}
	}
	\Else
	{
		\If{$\exists$ a free vertex $v \in V_r^2$ such that all $v' \in N_v^2 \cap V_r^2$ is either unoccupied or occupied with robots in Finished state}
		{
			Set $v$ as $r.Target$
		}
		\Else 
		{
			\textsc{Leadership\_Transfer}($r$)
		}
	}
	\caption{\textsc{MULTIND\_FindTarget}($r$)}
	\label{findtarget_app}
\end{algorithm2e}
\begin{algorithm2e} 
\If {$r.State$ is Leader}
	{
		\If {$r.Target$ is not set}
		{
			\If {$r.Entry$ is set and $r.Successor.Color=$ CONF3}
			{
				\textsc{MULTIND\_FindTarget}($r$) \& \textsc{Communicate}($r.Target$)\\
			
				$r$ clears $r.Target$ by clearing $r.Target.One$ and $r.Target.Two$\\
				Set $r.Color=$ WAIT-i \algorithmiccomment{(Wait color corresponding to Door-i)}\\
				\textsc{Packed\_State}($r$, $r.Successor$)\\
				
			}
		}
	}

\If{$r.State$ is Follower}
	{
		\If {$r.NextTarget.One$ is not set}
		{
			\textsc{Receive}($r.Predecessor.Target$)
		}
		\Else
		{
			\If{$r.Target$ is set}
			{
				\textsc{Communicate}($r.Target$)\\
				$r$ clears $r.Target$ by clearing $r.Target.One$ and $r.Target.Two$\\
				\textsc{Packed\_State}($r$, $r.Successor$)\\
			}
		}
	}

\If{$r.State$ is None \& $r.Color=$ON}
{
		\If {$r$ does not find any robot less than a distance of 2 hops}
		{ 
			Change $r.State$ to Leader\\
		   Set $r.Color=$ Wait-i\algorithmiccomment{(Wait color corresponding to the Door-i)}\\
			\textsc{MULTIND\_FindTarget}($r$) \& Set direction of movement as $r.Entry$\\
		 Set $r.Color=$ MOV \&
		Move to $r.Target$\\
		 Set $r.Color=$ WAIT-i\\
		}
		\Else
		{
			 $r$ chooses the nearest robot as $r.Predecessor$\\
			Change $r.State$ to Follower\\
			\If{$r.Predecessor.Color=$ WAIT-i \& $r.Color=$ON}
			{
				Set $r.Color=$ CONF3
			}		
	}
}
		\caption{MULTIND}
		\label{multind_app}
	\end{algorithm2e}

\subsection{Example of MULTIND} We use an example to illustrate the execution of the MULTIND algorithm.
Consider the graph with respective Door positions in Fig.~\ref{fig:multindflow1_app}.
Door 1 ranks higher in the hierarchy than Door 2, which ranks higher than Door 3.
Assume that all the robots are activated in each epoch.
The positions of the robots in the next epoch are shown in Fig.~\ref{fig:multindflow2_app}.
Each Leader robot makes a move.
Observe that Leader-1 and Leader-2 are within the visibility range of each other.
Similarly, Leader-2 and Leader-3 are within the visibility range.
Leader-2 is dominated by Leader-1 and hence, its target vertex should not lie on the path to Leader-1.
A similar argument applies for Leader-3 which is dominated by Leader-2.
Fig.~\ref{fig:multindflow3_app} shows the positions of the robots after the next epoch.
All the Leaders are now stuck and have to transfer their Leadership.
The result of that is shown in Fig.~\ref{fig:multindflow4_app}.
Leader-2 has no possible target vertices and has to transfer its leadership again.
Leader-3 is dominated by Leader Leader-2.  In the next epoch shown in Fig.~\ref{fig:multindflow5_app}, Leader-1 and Leader-3 make a move.
Leader-2 is at its Door.
In the next epoch, Leader-2 goes into \emph{Finished} state.
Similarly, in the next epoch, Leader-1, and Leader-3 transfer their leaderships which is shown in Fig.~\ref{fig:multindflow6_app}.
Leader-3 makes one final move in Fig.~\ref{fig:multindflow7_app}.
\begin{figure}
    \label{multindfig_app}
	\centering
	\begin{subfigure}{.45\textwidth}
	\centering
	\includegraphics[scale=0.07]{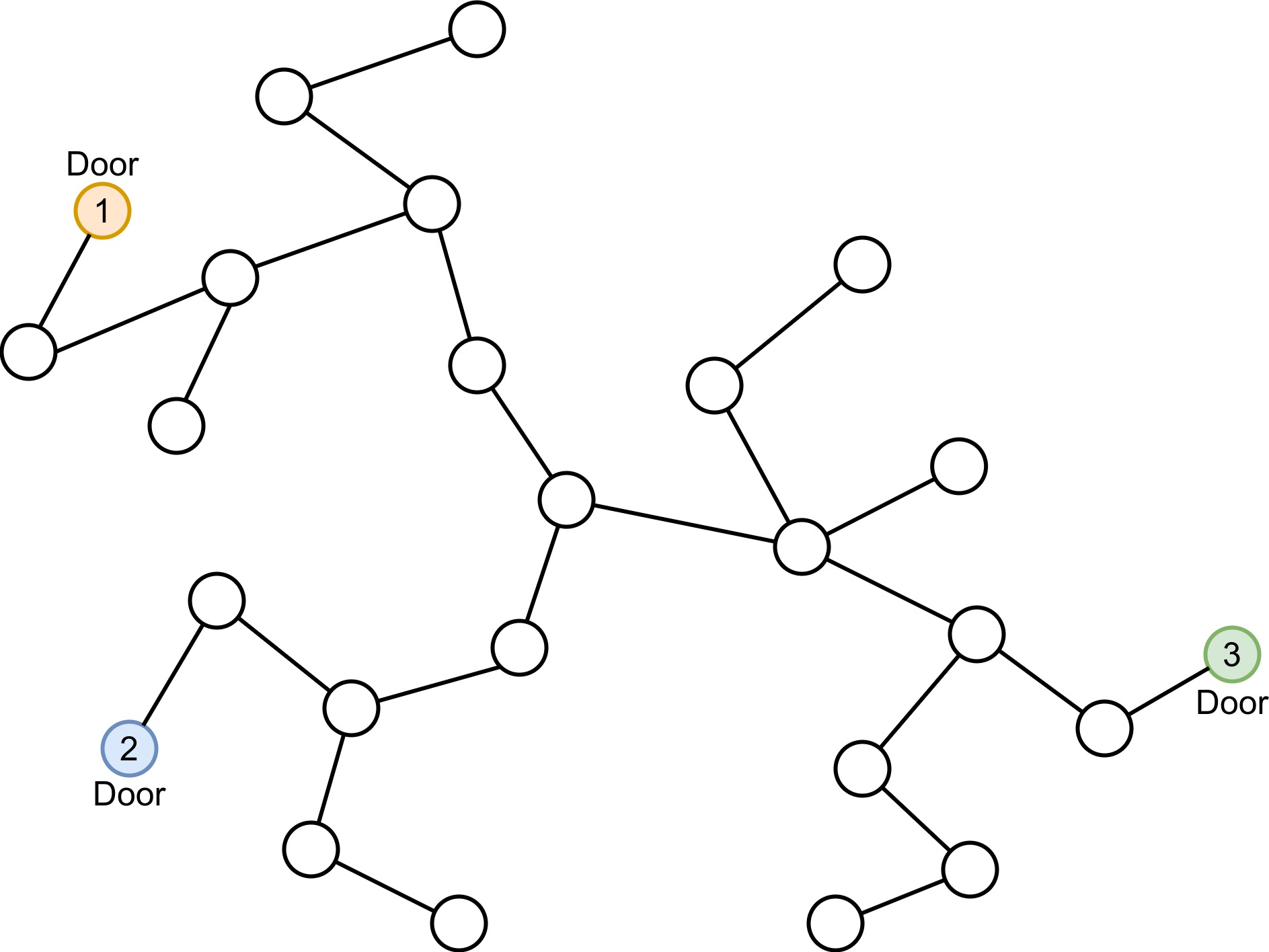}
	\caption{Location of the Doors and initial setting.}
	\label{fig:multindflow1_app}
	\end{subfigure}\hfill
	\begin{subfigure}{.45\textwidth}
		\centering
		\includegraphics[scale=0.07]{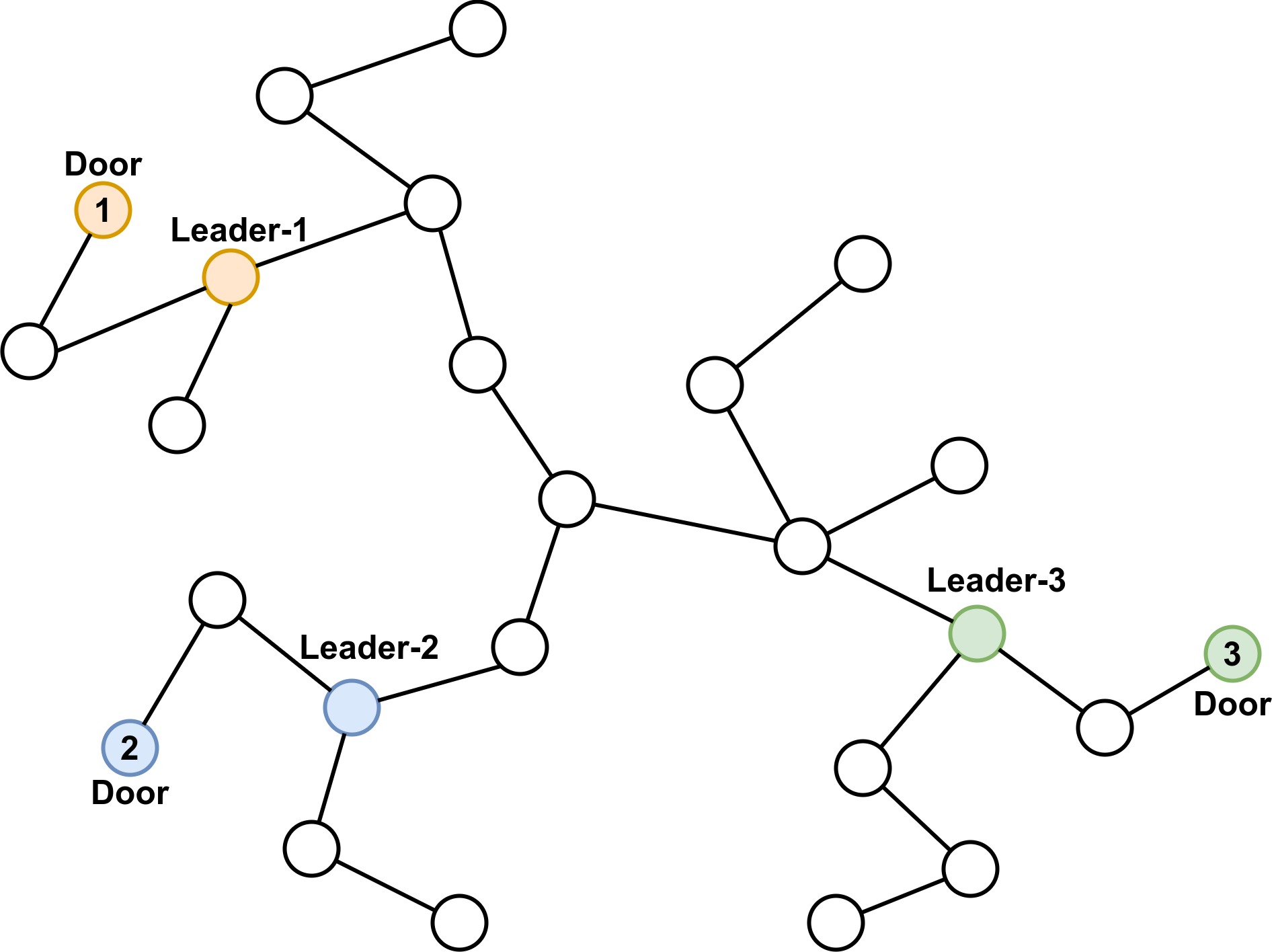}
		\caption{All the Leaders make a move.}
		\label{fig:multindflow2_app}
	\end{subfigure}~
	
	\begin{subfigure}{.45\textwidth}
		\centering
		\includegraphics[scale=0.07]{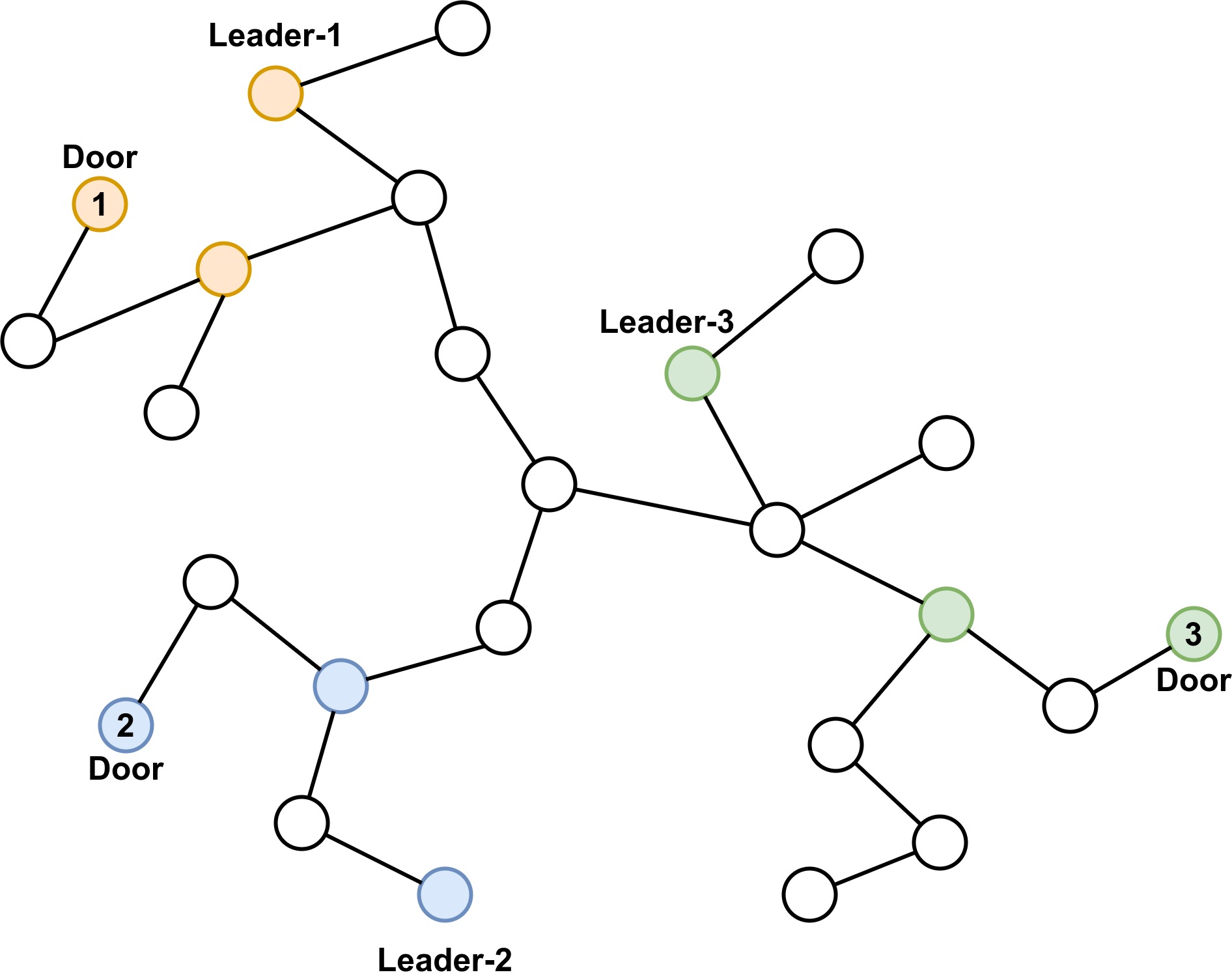}
		\caption{Leader-2 and Leader-3 avoid moving toward their dominating robot.}
		\label{fig:multindflow3_app}
	\end{subfigure}\hfill
	\begin{subfigure}{.450\textwidth}
			\centering
		\includegraphics[scale=0.07]{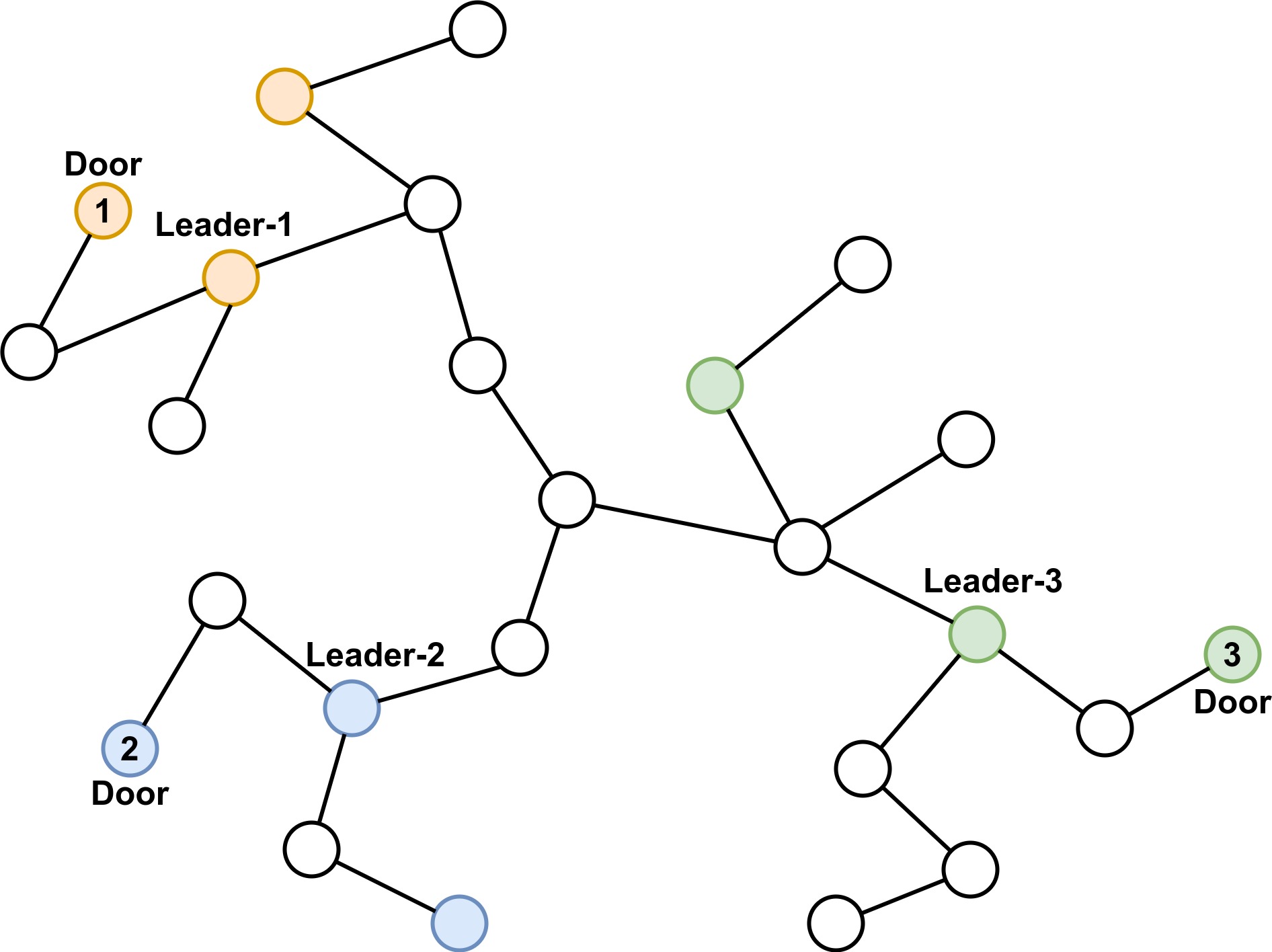}
		\caption{Leadership transfer occurs.}
		\label{fig:multindflow4_app}
	\end{subfigure}
	
	\noindent
	\begin{subfigure}{.45\textwidth}
			\centering
		\includegraphics[scale=0.07]{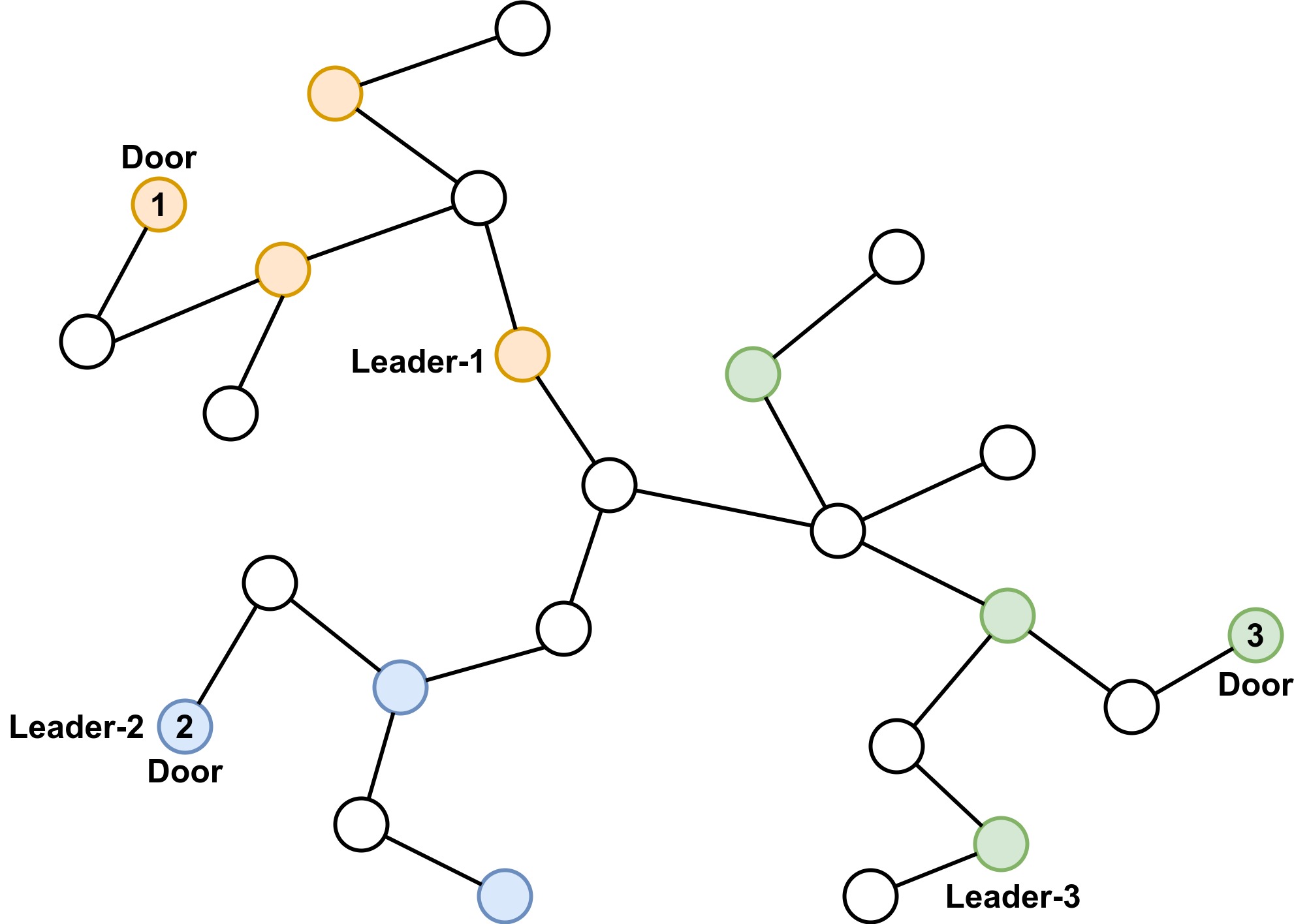}
		\caption{Leader-1 and Leader-3 make a move.
Leader-3 avoids moving toward Leader-1.}
		\label{fig:multindflow5_app}
	\end{subfigure}\hfill
	\begin{subfigure}{.45\textwidth}
		\centering
		\includegraphics[scale=0.07]{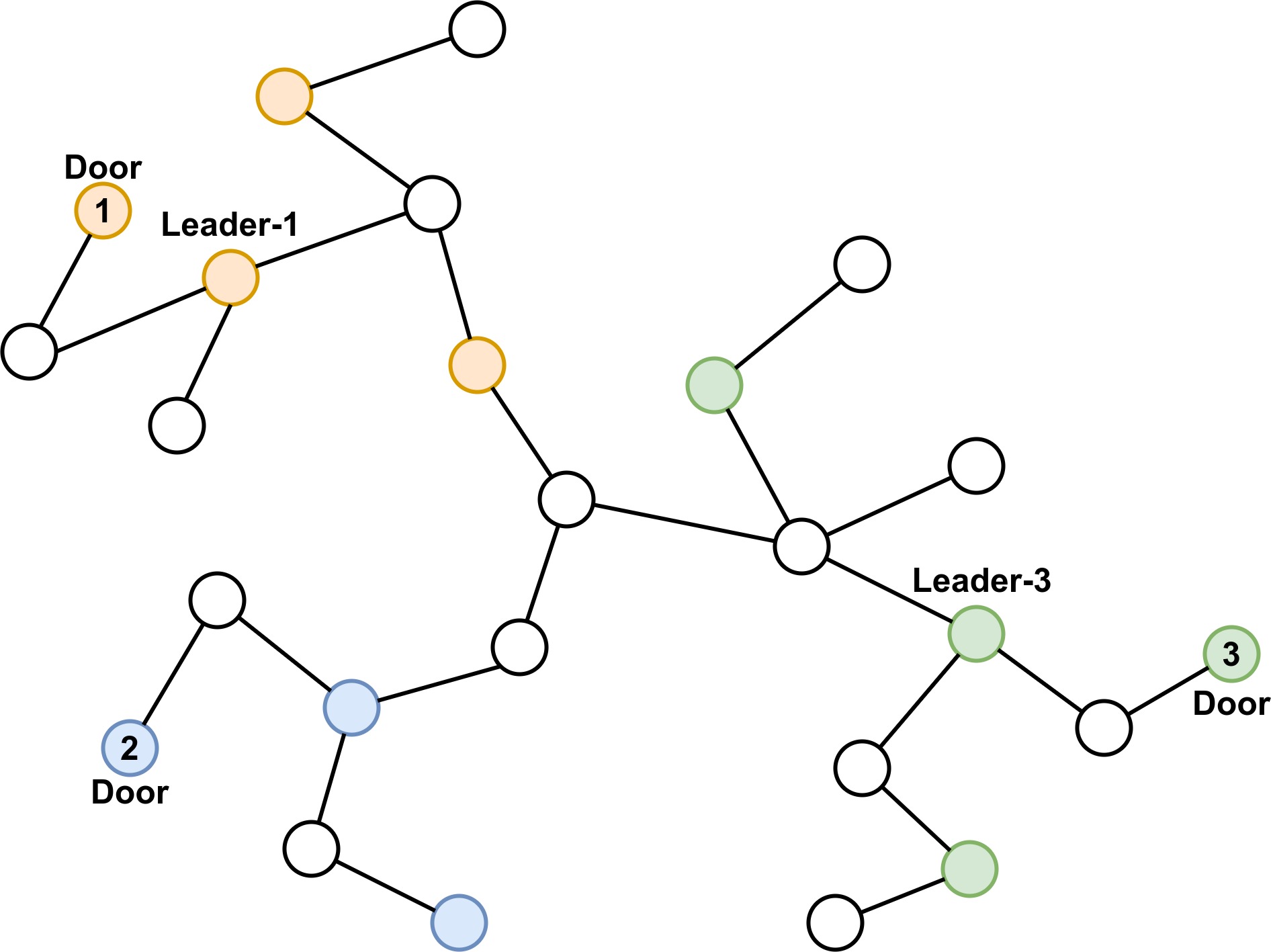}
		\caption{Leadership transfer occurs.}
		\label{fig:multindflow6_app}
	\end{subfigure}

	\begin{subfigure}{.45\textwidth}
		\centering
		\includegraphics[scale=0.07]{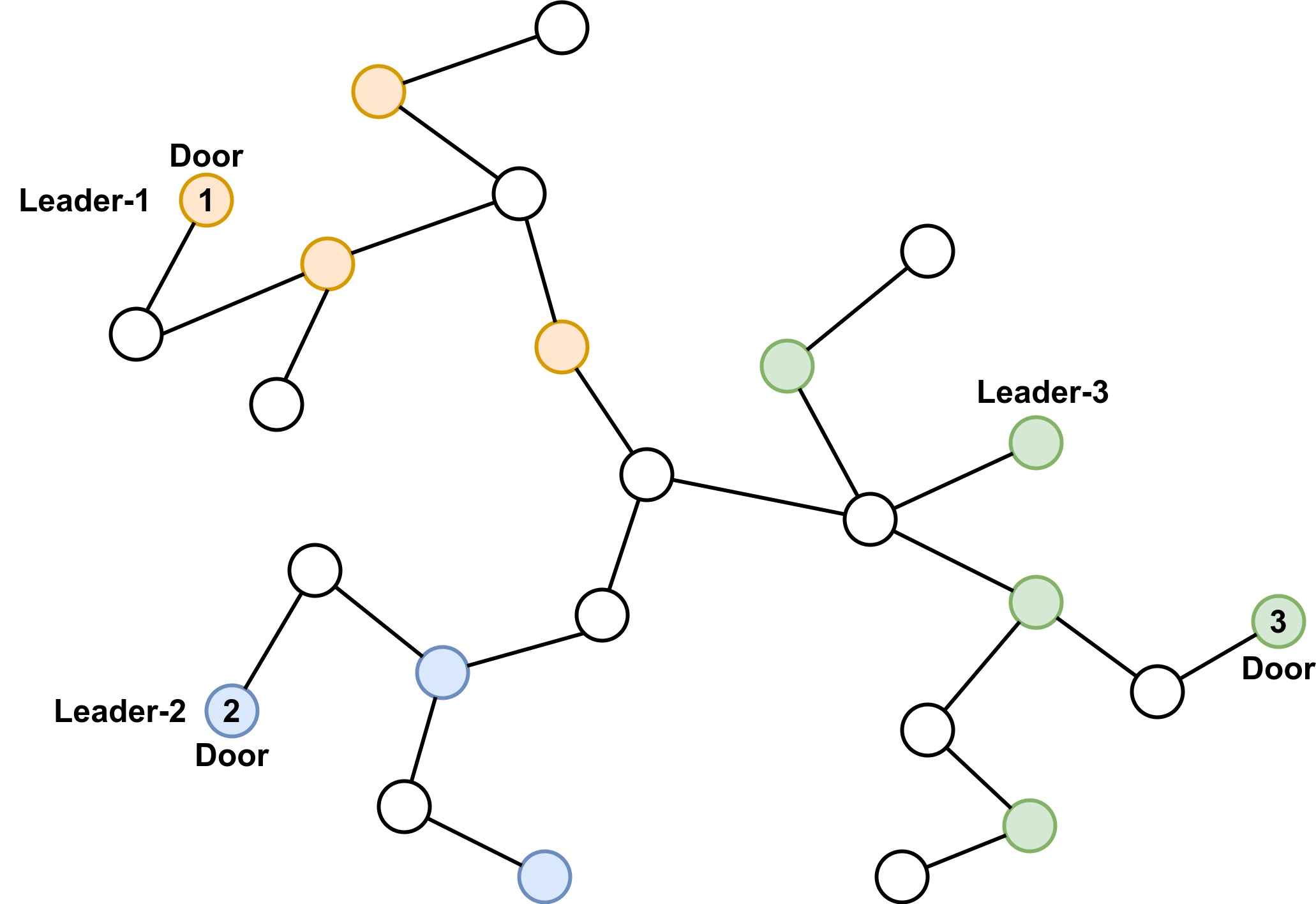}
		\caption{Leader-3 makes one final move.}
		\label{fig:multindflow7_app}
	\end{subfigure}
	\caption{An example execution of the MULTIND algorithm}
\end{figure}

\subsection{Analysis of MULTIND Algorithm}

\begin{lemma}\label{lem4_app}
    A Leader robot always moves to and occupies free vertices.
\end{lemma}

\begin{proof}
Since the robots have a visibility range of 5 hops, a Leader robot can choose a target vertex such that none of its neighboring vertices are occupied by a robot that is not a Leader. Also two leaders do not occupy adjacent vertices.
    For a leader $r_L$, if there is another Leader $r'_L$ within $5$ hops visibility of $r_L$, one of them dominates the other. Let $r_L$ be the dominating robot. MULTIND algorithm ensures that $r'_L$ does not choose a vertex as its target in the path where it is being dominated. So, either $r'_L$ finds a target on some other path where is not getting dominated or transfers the leadership (in case of no such target is there). So, the dominating Leader $r_L$ can choose a free vertex as its target from the path $\mathcal{P}_5(r_L,r'_L)$. Hence, the robots cannot choose target vertices such that they are adjacent. \qed
\end{proof}

\begin{lemma}\label{lem5_app}
    The Robots do not collide.
\end{lemma}

\begin{proof}\label{lem5proof_app}
    The proof for collision avoidance within a chain is  same as Lemma \ref{lem2_app} for a single Door case.
    In the \emph{Look} phase of the robots, if a Leader robot encounters another Leader robot within its visibility range, one of the robots gets dominated and will choose a target such that it does not lie on any path (of maximum length of 5 hops) connecting the two Leader robots. Thus, any two Leader robots avoid collision with each other.
    A Leader robot avoids collision with robots in another chain by avoiding cutting through a chain. When a leader robot finds a free vertex $v$, it checks every vertices in  $N_v^2 \cap V_{r_L}^2$. If at least one of the vertex in $N_v^2 \cap V_{r_L}^2$ is occupied by a robot not in Finished state, the Leader does not move. This prevents the crossover of two chains. Hence, a leader cannot collide with a follower robot from another chain. \qed
\end{proof}

Consequently, we can state following corollary.
\begin{corollary}
    A chain does not cross itself. Moreover, two chains do not cross each other. 
\end{corollary}

\begin{lemma}\label{lem6_app}
    No two robots in the Finished state occupy adjacent vertices.
\end{lemma}

\begin{proof}\label{lem6proof_app}
    A robot goes into the Finished state only after becoming a Leader. According to Lemma~\ref{lem4_app}, a Leader only moves to and occupies free vertices. Hence, when a robot enters the Finished state, none of the vertices adjacent to it are occupied. \qed
\end{proof}

\begin{theorem}\label{thm:multind_app}
    Let $G$ be a connected graph $G$ with $k$ Doors. Algorithm MULTIND fills an MIS of vertices in $O(m^2)$ epochs of $G$ under an SSYNC scheduler, without collisions, by mobile luminous robots having the following capabilities: visibility range of 5 hops, persistent storage of $O(\log (\Delta + k))$ bits, and $\Delta + k + 7$ colors.
\end{theorem}

\begin{proof}\label{thm:multindproof_app}
    By Lemma~\ref{lem5_app} and Lemma~\ref{lem6_app}, the filled vertices in $G$ form a MIS and the filling is done without collisions. The robots require $O(\log (\Delta + k))$ bits of memory to store the following: \emph{State} (4 states: 2 bits), \emph{Target} (directions to the target vertex: $2  \ceil{\log \Delta}$ + $\ceil{\log k}$ bits), \emph{NextTarget} (directions to the vertex to which the robot has to move after the \emph{Target} vertex is reached: $2  \ceil{\log \Delta}$ bits).

    The colors used by the robots are $\Delta$ colors to show the directions to the target of the robot, that also acts as a special color to switch to the Finished state.
    Initially, when a robot is place for the first time at Door, it is colored with ON.
    There are $k$ numbers of WAIT colors - one for each Door.
    There are three additional colors (CONF, CONFC, and CONF2) for confirming that the robot saw the signaled direction and confirmations of the Predecessor or the Successor, one color MOV used during the movement, one color CONF3 to indicate that the chain is in Packed state and the OFF color.
    
    Consider a graph where all the $k$ vertices that are connected to Doors node form a clique.
    In this case, only the Leader corresponding to the highest color Door would occupy one of the nodes of the $k$-clique. In this particular case, only one of the Doors remain active, and robots at all other Doors would go into the Finished state. 
    So the multi-Door case would behave as a single Door case and thus replicate the time-bound for a single Door case. 
    Thus, from Theorem~\ref{thm1_app} it follows that, MULTIND solves the \emph{MIS Filling problem} in a graph with multiple Doors in $O$($m^{2}$) epochs.
    \qed
\end{proof}

\section{Discussion}
        \subsubsection{On the requirement of the colors:}
    We can assert that at least $\Delta$ colors are required. Since the movement of robots via a port number at a node is marked by a corresponding DIR color, if we have less than $\Delta$ colors, there exist two ports that are marked by the same color, thus resulting in a configuration with adjacent vertices occupied by robots.
    In the absence of $k$ colors corresponding to the $k$ Doors, it is impossible to have a winner among all leader robots, resulting in a collision.  While forming MIS, to avoid collision, we also need to ensure that a chain of robots should not cross or overlap itself. 
    
    While we need $\Delta$ colors to communicate the port numbers,
    we can minimize the number of colors at the cost of increasing round complexity. We can use Hoffman encoding to reduce the number of colors used for the port numbers from $\Delta$ to a constant number of colors. Now, a sequence of colors would represent a port number instead of a particular color and increases time complexity by a factor of the size of the largest encoding.
    
   \vspace{-0.6cm}
    \subsubsection{One hop vs Two hops movement:}Our model considers that a robot moves two hops in one LCM cycle. While moving, the color of the robot is set to MOV. However, we can easily avoid this $2$ hops movement of a robot by replacing the color MOV with two colors, MOV1 and MOV2. After a robot chooses its target, it sets its color to MOV1 for the first hop and then changes its color to MOV2 before reaching the target. 

    \begin{figure}[H]
        \centering
        \includegraphics[width=0.9\linewidth]{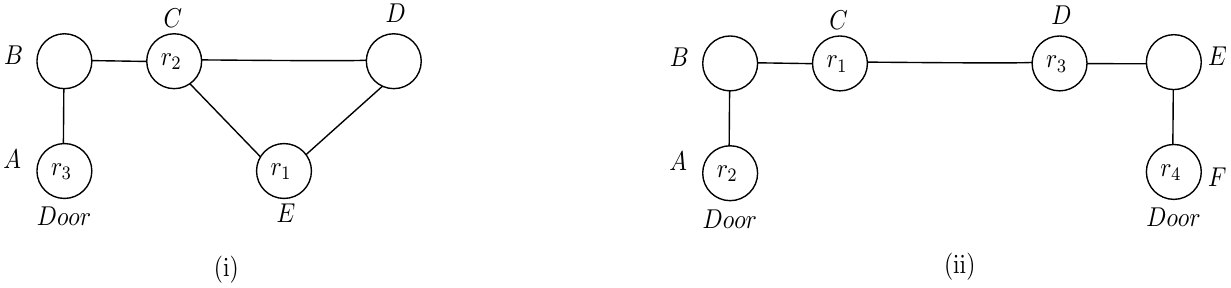}
        \caption{ Robots occupy adjacent nodes for (i) a single Door and visibility range two; (ii) multiple Doors and visibility range four.}
        \label{fig:vis31_app}\vspace{-1em}
    \end{figure}
    \subsubsection{Minimality of visibility range:} For the single Door case, a robot having a visibility range of two fails to avoid placing robots in adjacent nodes. Consider a graph as shown in Fig.~\ref{fig:vis31_app}(i). Initially, a robot appears at the Door vertex $A$; then it moves to $C$ a vertex two hops away. If the robots only have a visibility range of two, then $r_1$ can go to $E$ without realizing that $E$ and $C$ are connected, resulting in a configuration that is not an independent set. With a visibility range of two, a robot cannot determine whether a robot is present at the neighbor of the target vertex. Hence we need a visibility range of three for a single Door case.

    Similar to the single Door case, consider the graph in Fig.~\ref{fig:vis31_app}(ii) for the multi-Door case. If the robots have a visibility range of four, then the robots at $A$ and $F$ may simultaneously move to occupy $C$ and $D$ and result in a configuration with robots occupying adjacent vertices. 
    
\section{Conclusion}\label{sec:conclusion_app}

In this paper, we presented and analyzed two algorithms for solving two flavors of the problem of filling a  maximal independent set of vertices in an arbitrary connected graph using luminous mobile robots.
The first algorithm IND for graphs with a single Door works under an asynchronous scheduler. It uses robots having three hops of visibility range, $\Delta + 8$ number of colors, and $O(\log \Delta)$ bits of persistent storage and solves the problem in $O(n^2)$ epochs.
The second algorithm, MULTIND, works in graphs with $k~( > 1)$ Doors. It forms an MIS under a semi-synchronous scheduler using robots with five hops of visibility range, $\Delta + k + 7$ number of colors and having $O(\log (\Delta + k))$ bits of persistent storage, completing in $O(n^2)$ epochs. It is open to extending this algorithm to the generalized asynchronous scheduler. 
The model of attaching robot splitting Doors to a graph is a new direction in multi-robot coordination problems, and many other graph problems can be explored under the same model.

\bibliographystyle{splncs04}

\bibliography{mybibliography}

\end{document}